\documentclass[final,3p,times,authoryear]{elsarticle}

\usepackage{graphicx} 
\usepackage{amsmath,amssymb,amsfonts}

\usepackage{bm}

\usepackage{xcolor}
\usepackage{float}

\usepackage{amsthm}
\theoremstyle{plain}
\newtheorem{theorem}{Theorem}

\theoremstyle{definition}

\usepackage{makecell}
\usepackage{algorithm}
\usepackage{algpseudocode}

\usepackage{varwidth}
\usepackage{xcolor} 
\usepackage{hyperref}

\usepackage{comment}

\usepackage{enumitem}

\journal{Systems and Control Letters}

\begin{document}

\begin{frontmatter}

\title{Optimizing the Network Topology of a Linear Reservoir Computer}

\author[a]{Sahand Tangerami}
\author[b]{Nicholas A. Mecholsky}
\author[c,d]{Francesco Sorrentino}

\address[a]{K. N. Toosi University of Technology, Tehran, Iran}
\address[b]{Catholic University of America, Washington, D.C., United States}
\address[c]{University of New Mexico, Albuquerque, NM 87131, United States}
\address[d]{Max Planck Institute for the Physics of Complex Systems, 01187 Dresden, Germany}

\begin{abstract}
Machine learning has become a fundamental approach for modeling, prediction, and control, enabling systems to learn from data and perform complex tasks. Reservoir computing is a machine learning tool that leverages high-dimensional dynamical systems to efficiently process temporal  data for prediction and observation tasks.
Traditionally, the connectivity of {the network that underlies a reservoir computer (RC)} is generated {randomly}, lacking a principled design. Here, we focus on optimizing the {connectivity} of a linear RC to improve its performance and interpretability, which we achieve by decoupling the RC dynamics into a number of independent modes. We then proceed to optimize each one of these modes to perform a given task, which corresponds to selecting an optimal RC connectivity in terms of a given set of eigenvalues of the RC adjacency matrix. Simulations on networks of varying sizes show that the optimized RC significantly outperforms randomly constructed reservoirs in both training and testing
phases and often surpasses nonlinear reservoirs of comparable size. 
This approach provides both practical performance advantages and theoretical guidelines for designing efficient, task-specific, and analytically transparent RC architectures. 
\end{abstract}

\begin{keyword}
Machine learning, networks, optimization, reservoir computer.

\end{keyword}

\end{frontmatter}

\section{Introduction} 
Machine learning has emerged as a key approach for modeling and control, allowing systems to extract patterns from data and perform complex tasks with high accuracy \cite{Zhao_2025, BENSOUSSAN2022531, Hashemi2023}. However, training in machine learning can be computationally demanding. Reservoir computing (RC), a machine learning tool within the broader class of recurrent neural networks, addresses this challenge by keeping the internal reservoir connectivity fixed and training only the readout layer \cite{jaeger2001echo}. This allows training to be performed as a simple regression operation, as opposed to the more computationally intensive back-propagation.
 RC has been shown to be
highly successful in processing temporal and sequential data, especially in prediction and observation tasks \cite{10659224, 10816480, PhysRevLett.120.024102, 9664461,del2021reservoir,nathe2023reservoir}.

In typical implementations, reservoir computers employ nonlinear node dynamics;  {moreover, network topology, i.e., node-to-node connectivity of the reservoir, is chosen  randomly} \cite{Yan2024}. 
Although this random construction often yields satisfactory results in many tasks, it lacks a principled foundation and is not explicitly optimized to perform  specific tasks of interest \cite{Dale2021}. {This points out to an inherent trade-off between performance and repurposability, i.e., one can either use a reservoir with a random topology and expect it to perform different tasks well; or can pick a specific topology that is optimized to perform a specific task, but not others. In the rest of this paper, we focus on optimizing the structure  of {an} RC to perform specific tasks of interest.}

Furthermore, the presence of nonlinearities complicates the analysis of RC internal dynamics, making it difficult to characterize key properties such as memory capacity, stability, and input-output relationships  \cite{10.1162/neco_a_01770, 10221724}. Consequently, nonlinear RCs often function as black boxes, providing limited insight into how information is processed or stored. {Alternatively, one may opt for a linear reservoir computer (LRC), i.e., an RC whose dynamics is linear, which can significantly improve interpretability and analytical tractability \cite{Gauthier2021}.} 
Although linear reservoirs may exhibit lower performance in general, they can be highly effective for tasks with well-understood linear or frequency-driven structures, offering a valuable platform for developing theory-guided, efficient, and generalizable RC systems. 
Our approach in this paper adopts linear reservoir dynamics, which we will show allows systematic optimization of the network topology to enhance performance on specific tasks of interest.

In what follows, we consider an LRC performing an observation task, in which the reservoir is trained to reconstruct an output signal from knowledge of an input signal  \cite{10.1063/1.4979665}.
Our objective is to optimize the network topology of the LRC, based on knowledge of the spectral characteristics of the input and output signals. Specifically, we formulate an optimization problem to determine the optimal set of eigenvalues of the adjacency matrix that provides the connectivity between reservoir nodes, which we show leads to a drastic reduction of the training {and testing errors}. By leveraging frequency-domain insights, we identify reservoir topologies that maximize performance and interpretability.
 Our results show that the optimized LRC not only outperforms randomly constructed LRC{s} but also often achieves better performance than nonlinear RCs of comparable size.
\section{System Model and Reservoir Dynamics}
%

We take two signals which we know to be dynamically related to one another, $\hat{u}(t)$ and $\hat{y}(t)$. Fundamentally, our goal is to `observe' or `learn'  $\hat{y}(t)$ from knowledge of $\hat{u}(t)$. 
To this end, we extract a set of $K$ dominant frequencies that are simultaneously present in both $\hat{y}(t)$ and $\hat{u}(t)$, and construct the signals $y(t)$ and $u(t)$ that only contain those frequencies \cite{Butschek:22}. Hence, $u(t)$ and $y(t)$ can be expressed as:
\begin{align}\label{input_signal}
    u(t) &= \sum\limits_{k=1}^K a_k \cos\left(\omega_kt\right)\\
\label{output_signal}
  		y(t)&= \sum\limits_{k=1}^K b_k \cos\left(\omega_kt+\phi_k\right).
\end{align}
where $\{\omega_k\}$ is a set of distinct frequencies, $a_k$ and $b_k$ are amplitudes, and $\phi_k$ are the phases of the training signal with respect to the input signal.
We then opt for the simpler problem of learning the training signal  ${y}(t)$ from knowledge of the input signal ${u}(t)$. 

To this end, we introduce {an LRC} with $N$ nodes, which obeys the dynamics,
\begin{equation}\label{coupled_RC}
    \mathbf{\dot r}(t) = \gamma\left[ -  \mathbf{r}(t) + A \mathbf{r}(t) + \mathbf{d} u(t)\right],
\end{equation}
where   $\mathbf{r}(t)=[r_1(t),r_2(t),...,r_N(t)]$ is the $N$-dimensional vector of the nodes' {states},  $\gamma $ is a positive constant, and  $ A$ is the square $N$-dimensional reservoir's adjacency matrix,  i.e., { $A_{ij}$ describes the strength of the coupling from node $j$ to node $i$, the mask vector $\mathbf{d}=[d_1,d_2,...,d_N]$, and $u(t)$ is the input signal. We sometimes simply refer to the number of nodes $N$ as the reservoir `size'. We introduce two assumptions: (i) the matrix $A$ is diagonalizable and (ii) the matrix $(A-I)$ is Hurwitz. As a result of (ii) and under the generic assumption of bounded inputs $\mathbf{u}$, the reservoir dynamics \eqref{coupled_RC} is stable.}

The nodes states and the training signal are discretized with a time step $\tau$ and collected in the {state} matrix $\Omega_r$ and in the training vector $\mathbf{y}$, respectively,
\begin{equation}
\Omega_r=
\left(\begin{array}{ccccc}
    r_1(\tau) & r_2(\tau) & ... & r_N(\tau) & 1\\
    r_1(2\tau) & r_2(2\tau) & ... & r_N(2\tau) & 1\\
    \vdots & \vdots & \vdots &  \vdots  & \\
    r_1(T\tau) & r_2(T\tau) & ... & r_N(T\tau) & 1\\
\end{array}\right), 
\quad \quad \mathbf{y}=\left[\begin{array}{c} y(\tau) \\ y(2 \tau) \\ ...\\ y(T \tau) \end{array} \right],
\end{equation}
where $T$ is the number of time steps. {Throughout this paper $\tau$ is fixed and equal $0.01$; therefore,  the duration of the training phase is completely determined by the number of steps $T$.}

 We then compute an approximation of the training signal, the best-fit signal $\mathbf{{y}}^{\text{fit}}$, as a linear combination of the columns of the {state} matrix $\Omega_r$ in the coefficients of the output weights of the reservoir, which we carefully choose;  that is, we seek  the best-fit signal in the form $\mathbf{{y}}^{\text{fit}}=\Omega_r \pmb{\kappa}^*$, where here $\pmb{\kappa}^*$ is the vector of the output weights that minimize the  quantity 
$e(\pmb{\kappa} )= ||\Omega_r \pmb{\kappa} -\mathbf{y}||$ and
 $||.||$ denotes the Euclidean norm; that is, we compute the optimal output weights $\pmb{\kappa}^*= \operatorname*{arg\,min}_{\pmb{\kappa}}   e(\pmb{\kappa})$, which corresponds to setting
  $\pmb{\kappa}^*=  \Omega_r^{\dagger}\mathbf{y}$, where
${\Omega_r}^{\dagger} = \left(\Omega_r^T \Omega_r + \beta I \right)^{-1}\Omega_r^T$ and $\beta\geq 0$ is the ridge regression parameter  \cite{10.1007/978-3-540-87536-9_83}. With this prescription, we calculate the training error
\begin{equation} \label{norm_eror}
    \epsilon_r= \frac{1}{\sqrt{T}} ||\Omega_r \pmb{\kappa}^* -\mathbf{y}||.
\end{equation}

{Once we have learned the optimal output weights $\pmb{\kappa}^*$, we can also compute the testing error \cite{lukosevicius2009reservoir},
\begin{equation} \label{norm_eror}
    \check{\epsilon}_r= \frac{1}{\sqrt{T}} ||\check{\Omega}_r \pmb{\kappa}^* -\check{\mathbf{y}}||,
\end{equation}
where $\check{\Omega}_r$ and $\check{\mathbf{y}}$ are computed in the exact same way as ${\Omega}_r$ and ${\mathbf{y}}$, but over a different interval of time. In the rest of this paper, the duration of the testing phase is taken to be one third of that of the training phase, i.e, the number of samples during the testing phase is equal to $\lceil T/3 \rceil$ and the symbol $\lceil \cdot \rceil$ indicates ceil-rounding to the closest integer.}

To facilitate optimization of the linear reservoir, here we introduce a modal decomposition to transform the reservoir dynamics into independent modes that can be individually tuned.  {Since the matrix $A$ is diagonalizable,} 
it can be rewritten in the form $A = V\Lambda V^{-1}$, where $V$ is 
the matrix whose columns are the eigenvectors of $A$ and $\Lambda$ is a diagonal matrix with the {possibly complex} eigenvalues of $A$ along its main diagonal. We pre-multiply both sides of Eq.\ \eqref{coupled_RC} by $V^{-1}$, and by defining the vectors $\mathbf{q}(t) = V^{-1}\mathbf r(t)$ and $\mathbf{c} = V^{-1}\mathbf d$, the system transforms into the following form, 
\begin{figure} 
    \centering
    \includegraphics[width=0.5\linewidth]{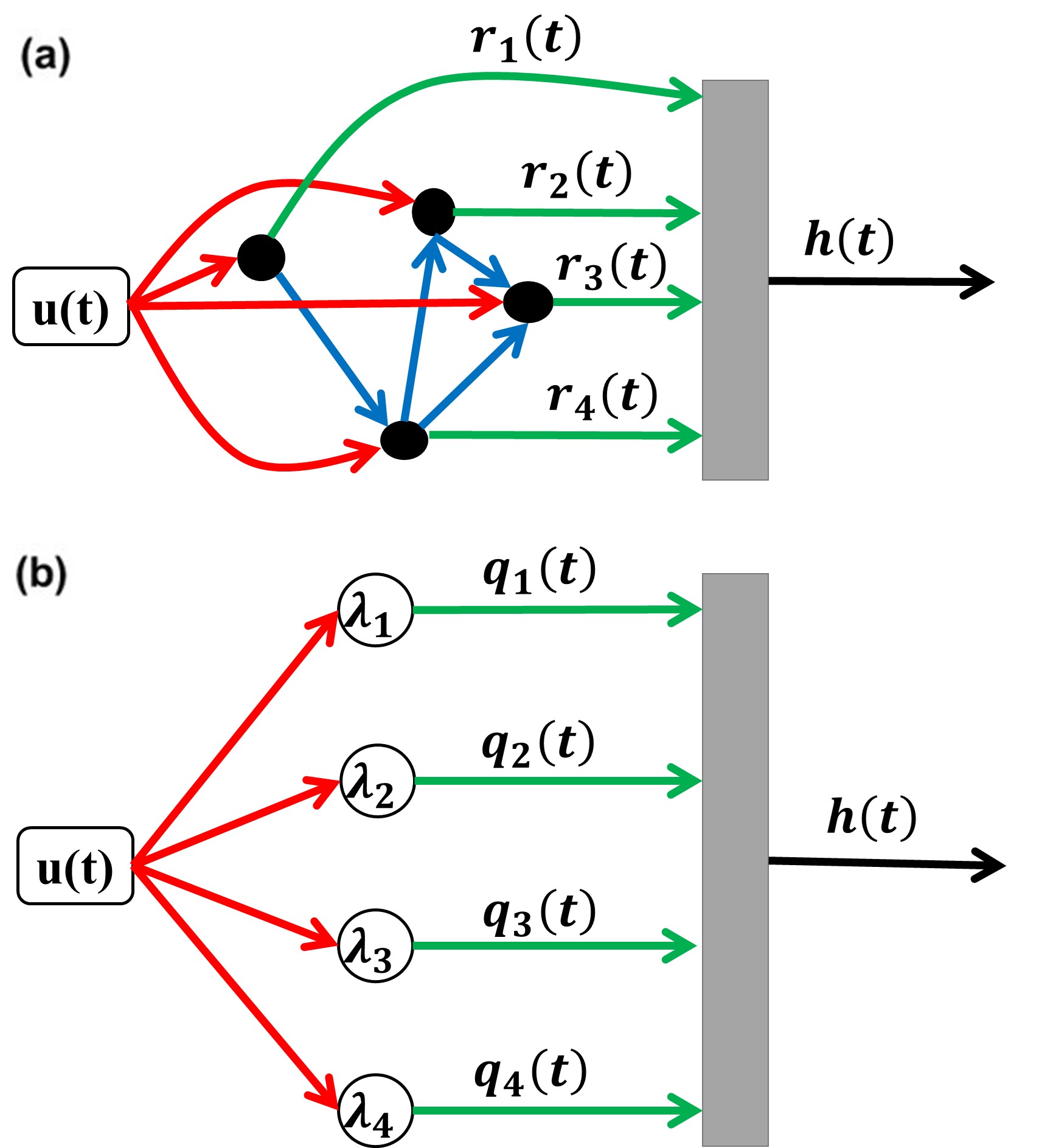}
    \caption{Two equivalent RC network topologies. a) A coupled reservoir computer. b) A decoupled reservoir computer.}\label{dec_coup_RC}
\end{figure}
\begin{equation}
\label{decoupled_RC}
    \mathbf{\dot q}(t) = \gamma\left[ - \mathbf{q}(t) + \Lambda \mathbf{q}(t) + \mathbf{c} u(t)\right],
\end{equation}
where $\mathbf{q}(t)$ is the $N$-dimensional vector of the decoupled reservoir {states}. Equation \eqref{decoupled_RC}  breaks up into a set of $N$ independent modes, 
\begin{equation}\label{individual reservior}
    {\dot q}_i(t) = \gamma(\lambda_i -1)q_i (t) + \gamma c_i u(t), \quad i=1,...,N,
\end{equation}
where $\lambda_i$ are the eigenvalues of the reservoir's adjacency matrix $A$. {In general, both the modes $q_i(t)$ and the eigenvalues $\lambda_i$ can be complex, but in the case in which $A$ is symmetric, they are real. }

Equations \eqref{decoupled_RC} and \eqref{individual reservior} provide a modal decomposition for the reservoir dynamics.
The implications of the modal decomposition of an RC are significant and are illustrated in Fig.\ \ref{dec_coup_RC}. Panel (a) of this figure shows a coupled reservoir, where nodes are interconnected, whereas panel (b) illustrates a decoupled reservoir computer, in which the reservoir is decomposed into independent modes without node-to-node connections. In particular, the fact that an RC can be decomposed into a number of independent modes (see Fig.\ \ref{dec_coup_RC}) has important implications and arguably several advantages  in terms of  RC robustness, scalability, and modularity.  {Next we briefly comment on each one of these advantages. Robustness arises from the fact that intentional or accidental removal of a single node (or of a set of the nodes) will not propagate to the others. Scalability is achieved since increasing the dimension of the reservoir simply corresponds to adding more individual nodes. This may have a considerable advantage in experimental setting, as it removes the need to wire different nodes in a complicated network structure. Finally, modularity is achieved when a complex structure is replaced by smaller, manageable, and interchangeable components.}
Another fundamental advantage is the possibility of more easily optimizing the structure of a reservoir, as it reduces the parameter space over which the optimization is performed from $N^2$ entries of the matrix $A$ to only $N$ eigenvalues. This is specifically what we do in this paper.

 The solution of Eq.\ \eqref{individual reservior}, is given by,
\begin{equation} \label{qq}
q_i(t) = \exp\Bigl( \gamma(\lambda_i - 1) t \Bigr) q_i(0) + \gamma c_i \int_0^t \exp\Bigl( \gamma (\lambda_i - 1)(t - \tau) \Bigr) u(\tau) \, d\tau
\end{equation}
where the first term on the left-hand side of \eqref{qq} is the free evolution, which by the assumption of stable reservoir dynamics goes to zero for large $t$. For large $t$, each mode $q_i$ differs from the others through the coefficient $\lambda_i$, while the particular modal amplitude is given by $c_i$. 

Similar to what done before, we can construct the {state} matrix for the decoupled reservoir,
\begin{equation} \label{omegaq}
\Omega_q=
\left(\begin{array}{ccccc}
    q_1(\tau) & q_2(\tau) & ... & q_N(\tau) & 1\\
    q_1(2\tau) & q_2(2\tau) & ... & q_N(2\tau) & 1\\
    \vdots & \vdots & \vdots &  \vdots  & \\
    q_1(T\tau) & q_2(T\tau) & ... & q_N(T\tau) & 1\\
\end{array}\right)
\end{equation}
and compute the error
\begin{equation} \label{es_errq}
    \epsilon_q = \frac{1}{\sqrt{T}} ||\Omega_q \pmb{\kappa}_q^* -\mathbf{y}||,
\end{equation}
where 
$\pmb{\kappa}_q^*= \Omega_q^{\dagger}\mathbf{y}$.

The decoupled reservoir is easier to handle since it consists of independent rather than dependent modes. An important step is to understand the relationship between the error $\epsilon_r$ for the coupled reservoir of Eqs.\ \eqref{coupled_RC}  and the error $\epsilon_q$ for the decoupled reservoir of Eqs.\ \eqref{decoupled_RC}. This relationship is formally established in Theorem 1.

\begin{theorem}
{Assume the matrix $A$ is diagonalizable.} Let a coupled linear reservoir and a decoupled linear reservoir be described by Eqs.\ \eqref{coupled_RC} and \eqref{decoupled_RC}, respectively. Then, if either $A$ is normal or $\beta=0$, we have that 
$\epsilon_q=\epsilon_r$.
\end{theorem}
\begin{proof}[Proof]
We construct the matrix 
\[
W =
\begin{bmatrix}
V & \mathbf{0}_{N\times 1} \\
\mathbf{0}_{1\times N} & 1
\end{bmatrix},
\] where $V$ is a generic invertible matrix $V$. The inverse of $W$ is equal to
\[
W^{-1} =
\begin{bmatrix}
V^{-1} & \mathbf{0}_{N\times 1} \\
\mathbf{0}_{1\times N} & 1
\end{bmatrix}.
\]
We recall that  $\mathbf{q}(t) = V^{-1}\mathbf r(t)$. Thus we have $\Omega_q=\Omega_r W^{-1}$ and $\Omega_q^T= {W^{-1}}^T \Omega_r^T$. By plugging in the expression for the optimal $\pmb{\kappa}_q^*=\Omega_q^\dagger \textbf{y}$ in Eq. (11), we obtain 
\begin{equation}
\begin{split}
  {\sqrt{T}}  \epsilon_q &=  ||(\Omega_q (\Omega_q^T \Omega_q + \beta I )^{-1} \Omega_q^T - I)\mathbf{y}|| \\
    &= ||(\Omega_r W^{-1} ({W^{-1}}^T\Omega_r^T \Omega_r W^{-1}+ \beta I)^{-1} {W^{-1}}^{T} \Omega_r^T - I)\mathbf{y}|| \\
    &= ||(\Omega_r W^{-1} W (\Omega_r^T \Omega_r + \beta I)^{-1} W^{T} {W^{-1}}^{T} \Omega_r^T - I)\mathbf{y}||\\
    &= ||(\Omega_r (\Omega_r^T \Omega_r + \beta W^T W)^{-1}   \Omega_r^T - I)\mathbf{y}|| \\
    &=\sqrt{T} ||(\Omega_r \Omega_r^{\dagger \beta,W}-I) \mathbf{y} ||, 
    \end{split}
\end{equation}
where the operation $\Omega_r^{\dagger \beta,W}=(\Omega_r^T \Omega_r + \beta W^T W)^{-1}   \Omega_r^T $ is similar to a ridge regression but with the ridge regression parameter $\beta$ replaced by the term $\beta W^T W$, which also depends on the particular choice of the matrix $W$.

In general we can conclude that $\epsilon_q=\epsilon_r^W$, where the error $\epsilon_r^W$ is computed by performing the alternative regression that uses $\Omega_r^{\dagger W}$ in place of $\Omega_r^{\dagger}$.
We also see that there are two cases under which $\epsilon_q=\epsilon_r$: (i) when the matrix $A$ is normal, which implies $W^T W=I$ and (ii) when $\beta=0$. Also, for small $\beta$, $\epsilon_q \simeq \epsilon_r$.

\end{proof}
In the rest of this paper we focus on the decoupled reservoir of Eqs.\ \eqref{decoupled_RC} since it is simpler to work with and it yields the same training and testing errors as the coupled reservoir (Theorem 1.)
Notwithstanding the benefits of the decoupled reservoir, the dimension of the {state} matrix remains large, as it scales with the number of time steps. This dimension can be significantly reduced by transforming the problem into the frequency domain, where the matrix size depends only on the number of input frequencies $K$, which in practice is considerably smaller than the number of time steps $T$. To this end, we reformulate the problem in the frequency domain by applying the Laplace transform to each mode and expressing it through its corresponding transfer function. 

{In what follows, proceed under the assumption that the matrix $A$ has a real spectrum and that the modes $q_i(t)$ are real.}
The transfer function of each mode in Eq.\ \eqref{qq} can be expressed as follows,

\begin{equation} \label{transfer_f}
  		T_i(j\omega)=\frac{\gamma c_i}{j\omega-\gamma(\lambda_i - 1)}.
\end{equation}
%
\subsection{Generic Relation Between the Error In The Time and the Frequency Domain}
If we remove the transient component, which corresponds to the first term on the right-hand side of Eq.\ \eqref{qq}, the time-domain solution for the dynamics of each node in the decoupled reservoir is given by: 
\begin{equation} \label{larget}
  		q_i(t)= \sum\limits_{k=1}^K a_kM_{ik}\cos\left(\omega_kt+\theta_{ik}\right)
\end{equation}
where $a_k$ denotes the amplitude of the input signal,  $\omega_k$ is the common  frequency between the input and output signals,  $M_{ik}$ and $\theta_{ik}$ represent the magnitude and phase of the reservoir, respectively, which are computed as follows: 
\begin{align}  
  		M_{ik}&= \left|\frac{\gamma c_i}{j\omega_k-\gamma(\lambda_i - 1)}\right|=\frac{\gamma |c_i|}{\sqrt{\omega_k^2+\gamma^2(\lambda_i - 1)^2}}\nonumber\\
        \label{mag_phase}
  		\angle{\theta_{ik}} &= \angle \frac{\gamma c_i}{j\omega_k-\gamma(\lambda_i - 1)}= \tan^{-1}\left(\frac{\omega_k}{-\gamma(\lambda_i-1)}\right)
\end{align}
%

Next, we take the Laplace transform of  Eq.\ \eqref{larget} and perform an alternative regression in the frequency domain. 
We note that the expression of Eq.~(\ref{larget}) can be interpreted as a linear combination of a set of basis functions, analogous to the representation of a vector in terms of its components. In this context, the basis functions—or \textit{primitives}—are given by  $\left\{ \cos(\omega_1t+\theta_{11}),\cos(\omega_2t+\theta_{12}),\ldots, \cos(\omega_Kt+\theta_{1K}) \right\}$ and the corresponding coefficients for  $q_1(t)$ are $\left\{ a_1M_{11}, a_2M_{12}, \ldots, a_KM_{1K} \right\}.$

The output of the reservoir computer is constructed as a linear combination of the internal states. 
By taking the Laplace transform of Eq. (\ \ref{larget}), we obtain,
\begin{equation} \label{s1}
\begin{aligned}
  		Q_i(s)&= \sum\limits_{k=1}^K a_k M_{ik} \Bigl[\cos(\theta_{ik}) \frac{s}{s^2+\omega_k^2}- \sin(\theta_{ik}) \frac{\omega_k}{s^2+\omega_k^2}\Bigr]
	  \end{aligned}
	\end{equation}
    
Thus the linear combination of $Q_1(s)$, $Q_2(s)$, $\hdots$, $Q_N(s)$ is equal to,
\begin{equation} \label{s2}
\begin{aligned}
    \sum\limits_{i=1}^N \tilde{\kappa}_i Q_i(s)&= \sum\limits_{k=1}^K \sum\limits_{i=1}^N a_k\Bigl[\tilde{\kappa} M_{ik} \cos(\theta_{ik}) \frac{s}{s^2+\omega_k^2} \Bigr]-  \sum\limits_{k=1}^K \sum\limits_{i=1}^N a_k\Bigl[ \tilde{\kappa}_i M_{ik} \sin(\theta_{ik}) \frac{\omega_k}{s^2+\omega_k^2}\Bigr]\\
    &=\sum\limits_{k=1}^K \sum\limits_{i=1}^N a_k\Bigl[ \tilde{\kappa}_i M_{ik} \cos(\theta_{ik}) \frac{s}{s^2+\omega_k^2}-   \tilde{\kappa}_i M_{ik} \sin(\theta_{ik}) \frac{\omega_k}{s^2+\omega_k^2}\Bigr]
    \end{aligned}
\end{equation}
At the same time, the Laplace transform of the output signal in Eq. (\ \ref{output_signal}) is, 
\begin{equation} \label{s3}
\begin{aligned}
  		Y(s)&= \sum\limits_{k=1}^K b_k\Bigl[\cos(\phi_{k}) \frac{s}{s^2+\omega_k^2} - \sin(\phi_{k}) \frac{\omega_k}{s^2 +\omega_k^2}\Bigr]
	  \end{aligned}
	\end{equation}

Thus if the frequencies are all nonzero and different from each other, the primitives will be linearly independent and the only way two linear combinations would be the same is if all of the coefficients are equal. 

The coefficients $\tilde{\kappa}_i$ that minimize the error $\tilde{\epsilon}_q=||\sum_i \tilde{\kappa}_i Q_i(s) - Y(s) ||$ between the two functions in  Eq.\ (\ref{s2}) and Eq.\ (\ref{s3}) are,
\begin{equation}
\pmb{\tilde\kappa}^*=\tilde{\Omega}_q^\dagger  \mathbf{B},
\end{equation}
%
%
where $\tilde{\Omega}_q$ is given by 
\begin{equation}
\tilde{\Omega}_q =
\left(
\begin{array}{cccc}
a_1 M_{11} \cos(\theta_{11}) & a_1 M_{21} \cos(\theta_{21}) & \cdots & a_1 M_{N1} \cos(\theta_{N1})\\
a_1 M_{11} \sin(\theta_{11}) & a_1 M_{21} \sin(\theta_{21}) & \cdots & a_1 M_{N1} \sin(\theta_{N1})\\
a_2 M_{12} \cos(\theta_{12}) & a_2 M_{22} \cos(\theta_{22}) & \cdots & a_2 M_{N2} \cos(\theta_{N2})\\
a_2 M_{12} \sin(\theta_{12}) & a_2 M_{22} \sin(\theta_{22}) & \cdots  & a_2 M_{N2} \sin(\theta_{N2})\\
\vdots & \vdots & \ddots & \vdots \\
a_K M_{1K} \cos(\theta_{1K}) & a_K M_{2K} \cos(\theta_{2K}) & \cdots & a_K M_{NK} \cos(\theta_{NK})\\
a_K M_{1K} \sin(\theta_{1K}) & a_K M_{2K} \sin(\theta_{2K}) & \cdots & a_K M_{NK} \sin(\theta_{NK})
\end{array}
\right)
\label{matrix_rep}
\end{equation}
and 
%
\begin{equation}
\begin{array}{cc}
\boldsymbol{\tilde{\kappa}} = 
\begin{pmatrix}
\tilde{\kappa}_1 \\
\tilde{\kappa}_2 \\
\vdots \\
\tilde{\kappa}_N
\end{pmatrix}
\quad\quad
&
\mathbf{B} =
\begin{pmatrix}
b_1 \cos(\phi_1) \\
b_1 \sin(\phi_1) \\
b_2 \cos(\phi_2) \\
b_2 \sin(\phi_2) \\
\vdots \\
b_K \cos(\phi_K) \\
b_K \sin(\phi_K)
\end{pmatrix}
\end{array}.
\end{equation}

The main benefits of computing $\pmb{\tilde\kappa}^*$, as opposed to computing $\pmb\kappa^*$, is that the matrix 
$\tilde{\Omega}_q$ in  Eq.\ \eqref{matrix_rep} has a number of rows equal to 2$K$, where $K$ is the number of frequencies, which is typically much lower than the number of rows $T$ of the matrix 
${\Omega}_q$ in Eq.\    \eqref{omegaq}.
This results in lower computational cost, a simplified interpretation and analysis, and more efficient calculations. While especially useful for periodic or narrowband inputs, an essential step is to understand the relationship between errors in the time and frequency domains, due to the inherently temporal nature of learning in reservoirs. Theorem 2 explicitly establishes this relationship.\color{black}

\begin{theorem}
Let $\Omega_q \in \mathbb{R}^{T \times N}$ be the  {state} matrix of a {linear} decoupled reservoir computer, and let $\pmb{y} \in \mathbb{R}^T$ be the corresponding training signal in the time domain. Define the time-domain readout weights $\pmb \kappa \in \mathbb{R}^N$ as:
        \begin{equation*}  
    {\pmb \kappa} = \Omega_q^{\dagger} {\pmb y},
    \end{equation*}
where as specified before, $\Omega_q^{\dagger}={\left(\Omega_q^T \Omega_q + \beta I \right)^{-1}\Omega_q^T}$. Let $\widetilde{\Omega}_q$ and $ \tilde{\pmb y}$ be the discrete Fourier transforms (DFT) of $\Omega_q$ and $\pmb{y}$, respectively. Then, the frequency-domain readout weights, 
    \begin{equation*} \label{kappa_rep}
\tilde{\pmb \kappa} = \widetilde{\Omega}_q^{\dagger} { \tilde{\pmb y}}
\end{equation*}
are {approximately} equal to the time-domain readout weights, i.e.,
    \begin{equation*} 
\pmb\kappa \approx \tilde{\pmb \kappa}.
\end{equation*}
and consequently, the regression error is preserved in the two domains, i.e., $\tilde{\epsilon}_q \approx \epsilon_q.$ 
\end{theorem}
\begin{proof}[Proof]
The detailed proof is provided in Section S2 of the supplementary material. {The errors are only approximate as there is always a discretization error. This equivalence improves with a closer discretization approximation.}
\end{proof}
Theorem 2 demonstrates that the error of the decoupled reservoir in the time domain is approximately equal to the error in the frequency domain. Combining the results of Theorems 1 and 2, we obtain the relations,

\begin{equation} \label{err_rel}
\tilde{\epsilon}_q \approx \epsilon_q = \epsilon_r 
\end{equation}
which shows that, regardless of whether the regression problem is posed in the coupled or decoupled form, and in the time or frequency domain, the resulting error is essentially the same. 

While $\epsilon_r, \epsilon_q$, and $\tilde{\epsilon}_q$  are useful for theoretical analysis, in our numerical experiments we use the Normalized Root Mean Squared Error (NRMSE) to compare the performance of different reservoir architectures.
 \begin{equation} \label{NRMSE}
\mathrm{NRMSE} = \frac{\|\textbf{y} - \textbf{y}^{\text{fit}}\|}{\|\textbf{y}\|},
\end{equation}
where $\textbf{y}^{\text{fit}}$ is the best-fit approximation of the target signal and  $\textbf y$ is the true (measured or reference) signal. 
\subsection{Optimization}
Having established the relationships between the errors \eqref{err_rel}, we now propose an optimization framework to minimize the error of the linear reservoir computer by identifying an optimal set of eigenvalues $\lambda_i$ in Eqs.\ \eqref{decoupled_RC} and \eqref{individual reservior}. The primary goal of this procedure is to find these eigenvalues along with the output weights $\pmb{\kappa}$ that jointly minimizes the training error, as we explain below. 

From \eqref{transfer_f} we see that each mode acts as a low-pass filter. Thus, we require the cut-off frequencies to exceed the maximum frequency present in the input signal. From Eq.\ \eqref{transfer_f}, each mode has cut-off frequency $\omega_{c,i} = \gamma (1-\lambda_i)$. Hence, we impose $\omega_{c,i} > \omega_{\max} = \max_k \omega_k$ for all $i$, which yields the condition
\begin{equation}
\label{condit}
    \lambda_i < 1 - \frac{\omega_{\max}}{\gamma}, \quad \forall i.
\end{equation}

Given that the input signal $u(t)$ and the output signal $y(t)$ are defined in Eq.\ \eqref{input_signal} and Eq.\ \eqref{output_signal}, respectively, we now formulate the following optimization problem, incorporating the constraint in Eq.\ \eqref{condit}. 
\begin{subequations}\label{optim}
\begin{align}
\label{eq:19a}
    \min_{\lambda_1, \lambda_2, \ldots, \lambda_N, \pmb{\kappa}} \quad & \tilde{\pmb{\epsilon}}_q^\top \tilde{\pmb{\epsilon}}_q +\beta_1\tilde{\pmb{\kappa}}^\top\tilde{\pmb{\kappa}}+\beta_2 \left( \frac{1}{H} \right)\\
    \text{s. t.} \quad  
    & \tilde{\pmb{\epsilon}}_q = \tilde{\Omega} \tilde{\bm{\kappa}} - \tilde{y}, \\
    & \tilde{y} = [b_1 \cos(\phi_1), b_1 \sin(\phi_1), \nonumber\\ 
    &\phantom{\tilde{y} =}b_2 \cos(\phi_2),
    b_2 \sin(\phi_2), \ldots, \nonumber \\ &\phantom{\tilde{y} =} b_K \cos(\phi_K), 
    b_K \sin(\phi_K)]^\top, \\
    & \tilde{\Omega} = \tilde{\Omega}_q, \\
    & M_{ik} = \left|\frac{\gamma c_i}{j\omega_k-\gamma(\lambda_i-1)}\right| , \\
    & \theta_{ik} = \angle \frac{\gamma c_i}{j\omega_k-\gamma(\lambda_i-1)}, \\
    & \lambda_i \leq 0, \\
  & H = \frac{N}{\sum_{\substack{j,z=1 \\ j \neq z}}^{N} \frac{1}{|\lambda_j - \lambda_z|}},\\
  & \omega_{max} + \gamma(\lambda_i-1)\leq0
\end{align}
\end{subequations}

The cost function Eq.\ \eqref{eq:19a} includes three terms. The first term represents the reservoir's error in the frequency domain and serves as the primary quantity to be minimized. High readouts {weights} often indicate overfitting and can amplify noise in the reservoir states or inputs, leading to poor generalization and reduced robustness. To mitigate this, a second term is included in the cost function that minimizes the norm of the readouts {weights}, where 
$\beta_1 \geq 0$ is a  constant controlling the weight of this penalty. 

{Some comments should be made regarding the third term.}
The third term is the reciprocal of the harmonic mean of the differences between the eigenvalues, which is lower when {any of} the differences in eigenvalues are high. It is included in the cost function to enforce distinctness among the eigenvalues, and $\beta_2 \geq 0$ is the weight that controls the strength of this term. {Our analysis shows that identical eigenvalues would produce a computational redundancy and reduce the effective dimension of the reservoir. It was found that without the heuristic term, sometimes a numerical issue would create eigenvalue copies. Considering that this is undesirable, a heuristic algebraic barrier was included to keep the eigenvalues from being identical. The harmonic mean is a simple algebraic way of doing this because of the denominator of the differences of eigenvalues. When the eigenvalues are not near each other, the harmonic term does not contribute significantly and when the eigenvalues come close in value, the harmonic term becomes a significant contribution to the objective function. It is equivalent to a barrier method.}

\section{Numerical Verifications}
We validate the theoretical results through numerical simulations. Reservoir computer simulations are carried out in both MATLAB and Julia. The optimization problem is numerically solved using Ipopt, a nonlinear optimization solver, in combination with the JuMP modeling framework in Julia, following Algorithm \ref{optalg}, and the best performance is stored as the optimal solution. The optimization procedure in Algorithm~\ref{optalg} is nonconvex; therefore, it is executed multiple times with different initial values of $\lambda_0$. As the number of reservoir nodes increases, the dimensionality of the optimization problem also grows, which makes the search more challenging and necessitates a larger number of initial points to reliably approach an optimal solution.

Two strategies are used in parallel to generate the initial guesses. The first strategy is a warm-start approach, in which Algorithm~\ref{optalg} is repeatedly executed using various initial values of the eigenvalues $\boldsymbol{\lambda}_0$, which we draw from a uniform distribution in the interval 
{$[-19,1]$}. For each selected $\lambda_0$, the optimization is run once, and the resulting optimized value is then used as the initial guess for the next iteration to escape local minima. This process continues until the number of initial points drawn from the uniform distribution reaches the predefined maximum number for the warm-start optimization.
\color{black}
The vector  $\boldsymbol{\kappa}_0$ is then obtained using ridge regression, while $\boldsymbol{M}_0$, and $\boldsymbol{\theta}_0$ are computed according to Eq.\ \eqref{mag_phase}. (The implementation code is made publicly available and can be downloaded  { \href{https://github.com/SahandTerami/LinearReservoirOptimization/tree/main}{here}}) 
\begin{algorithm}[H]
\caption{Reservoir Eigenvalue Optimization}
\begin{algorithmic}[1]

\Require $ N, K, \beta_1, \beta_2, \gamma, \mathbf{a}, \mathbf{b}, \mathbf{\Phi}, \boldsymbol{\omega}, \mathbf{c}, \boldsymbol{\lambda}_0, \boldsymbol{\kappa}_0, \boldsymbol{M}_0, \boldsymbol{\theta}_0, warm\_start\_chain$
\Ensure Optimized $\boldsymbol{\lambda}, \boldsymbol{\kappa}, \boldsymbol{M},\boldsymbol{\theta}$, lowest\_error

\State $\mathbf{W} \leftarrow \frac{1}{\boldsymbol{\omega}}$
\State $lowest\_error \leftarrow \infty$
\State $num\_iter \leftarrow warm\_start\_chain$ (To escape sub-optimal local minima)

\For{$iter = 1$ to $num\_iter$}
    \State Cost = $\min \sum_{i=1}^{K} W_i \epsilon_{q,i}^2 + \beta_1 \sum_{j=1}^{N} \kappa_j^2 + \beta_2 (1/H)$
    \State Add nonlinear constraints
    \State Set initial values for $\lambda, \kappa, M, \theta$ from the previous iteration or initial guess.
    
    \State Solve the optimization problem and find $\boldsymbol{\lambda}_{new}, \boldsymbol{\kappa}_{new}, \boldsymbol{M}_{new},\boldsymbol{\theta}_{new}$
    \If{solver converged successfully}
        \State $\epsilon =   \sqrt{\sum_i W_i \epsilon_{q,i}^2}$
        \If{$\epsilon <lowest\_error$}
            \State $lowest\_error \leftarrow \epsilon$
            \State Store current $\boldsymbol{\lambda}_{new}, \boldsymbol{\kappa}_{new},\boldsymbol{M}_{new}, \boldsymbol{\theta}_{new}$ as the best solution.
        \EndIf
    \Else
        \State Keep previous $\boldsymbol{\lambda}_{best}, \boldsymbol{\kappa}_{best}, \boldsymbol{M}_{best},\boldsymbol{\theta}_{best}$ as the current solution.
    \EndIf
\EndFor

\State \Return $\boldsymbol{\lambda}_{best}, \boldsymbol{\kappa}_{best}, \boldsymbol{M}_{best},\boldsymbol{\theta}_{best}, lowest\_error$

\end{algorithmic}  \label{optalg}
\end{algorithm}
\subsection{Optimal Reservoir Computer}
To demonstrate the effectiveness of the proposed optimization framework, three numerical examples are considered for a reservoir computer governed by Eq.~\eqref{coupled_RC}, with number of nodes $N=10$, $N=100$, and $N=200$. In the three examples, the reservoir is trained for $T=3000$ steps with parameters $\beta_1 = 10^{-7}$, $\beta_2 = 10^{-1}$, and $\gamma = 6$, and the reservoir states together with the corresponding errors are computed using both analytical and numerical methods. The matrix $A$ is taken to be the adjacency matrix of a weighted/unweighted Erdos-Renyi graph with its spectrum shifted to ensure the Hurwitz property, and the input mask vector $\textbf{c}$ has entries randomly drawn from a Normal distribution. We set a larger number of warm-start guesses as the reservoir size grows, namely the warm-start count is set to 50 for the $N=10$ reservoir, 400 for $N=100$, and 1500 for $N=200$. In all cases, the length of each warm-start chain is fixed at 40. Consequently, the total number of initial guesses is equal to $50 \times 40 = 2000$ for $N=10$, $400 \times 40 = 16{,}000$ for $N=100$, and $1400 \times 40 = 60{,}000$ for $N=200$.
The input and output signals are defined as follows:
\[
u(t) = 1.1\cos\left( t\right)+ 1.7\cos\left(3 t\right)+ 2.1\cos\left( 5t\right)
\]
\[
  		y(t)= 2.2\cos\left(t-0.5\right)+\cos\left(3t+0.9\right)+1.6\cos\left(5t+1.1\right).
\]

The reservoir states and their corresponding prediction errors are computed both prior to and following the optimization process. By comparing the pre- and post-optimization results, we assess whether the optimization successfully reduced the observation error. 

Figure\ \eqref{example_perf}  demonstrates the performance of the reservoir computer in the training and testing phases, respectively, before and after the optimization. As shown in the figure, prior to optimization, the reservoir was unable to learn the underlying structure of the signal resulting in a large approximation error. In contrast, after the optimization, the reservoir successfully captured the signal dynamics and reproduced the target output with significantly improved accuracy. \\
\begin{figure}
    \centering
    \includegraphics[width=1\linewidth]{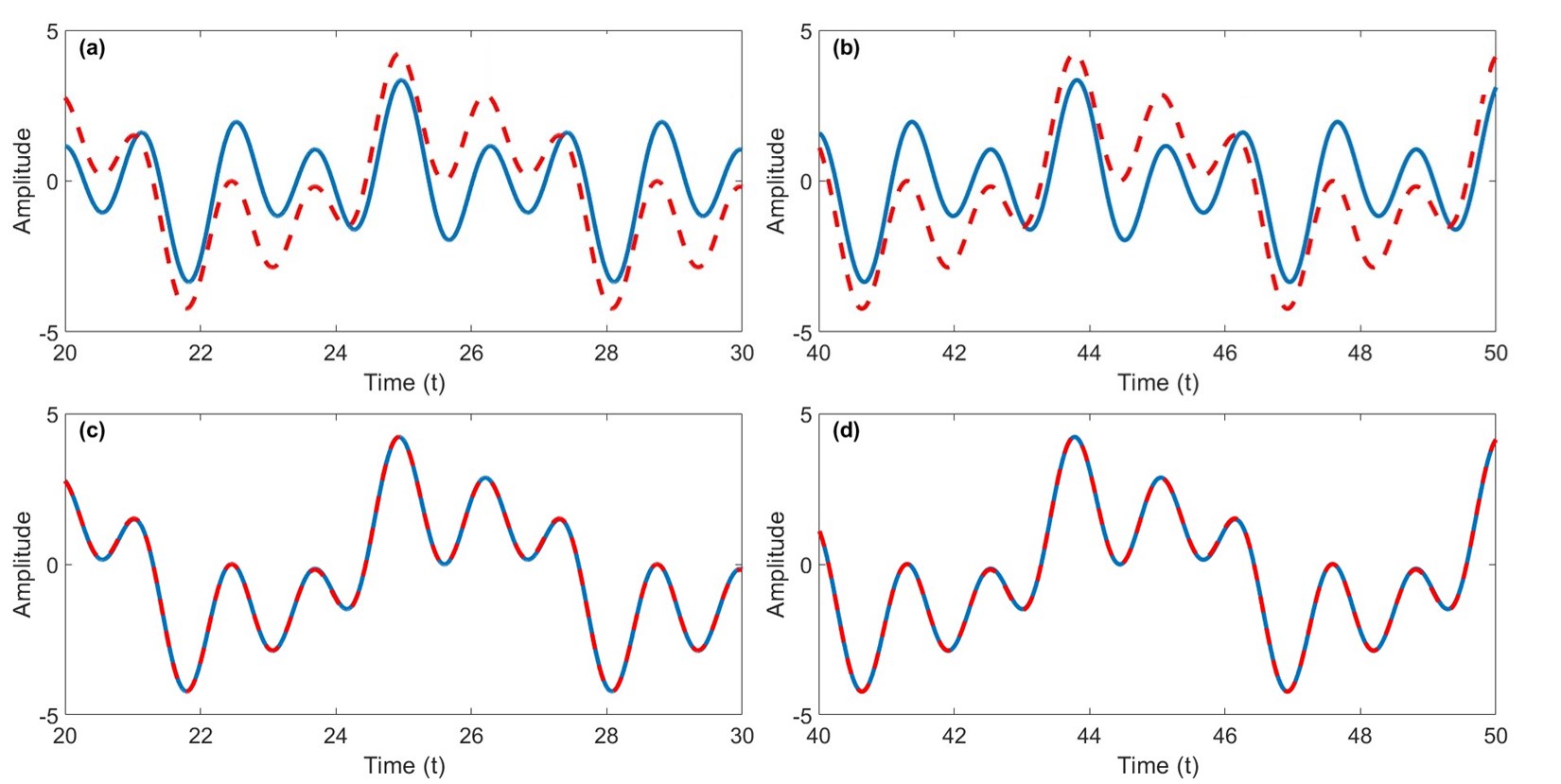}
    \caption{Performance of the 10-node reservoir. Red dashed curves represent the training signal, while blue curves correspond to the best fit produced by the reservoir: (a) Training before optimization, (b) Testing before optimization, (c) Training after optimization, (d) Testing after optimization. }
    \label{example_perf}
\end{figure}
The optimization method also scales to larger networks. Figure~\eqref{example_100_train} shows that, for a $N=100$ reservoir, the unoptimized network fails to learn the signal and exhibits large errors, while the optimized network accurately captures the dynamics and reproduces the target output.
\begin{figure}
    \centering
    \includegraphics[width=1\linewidth]{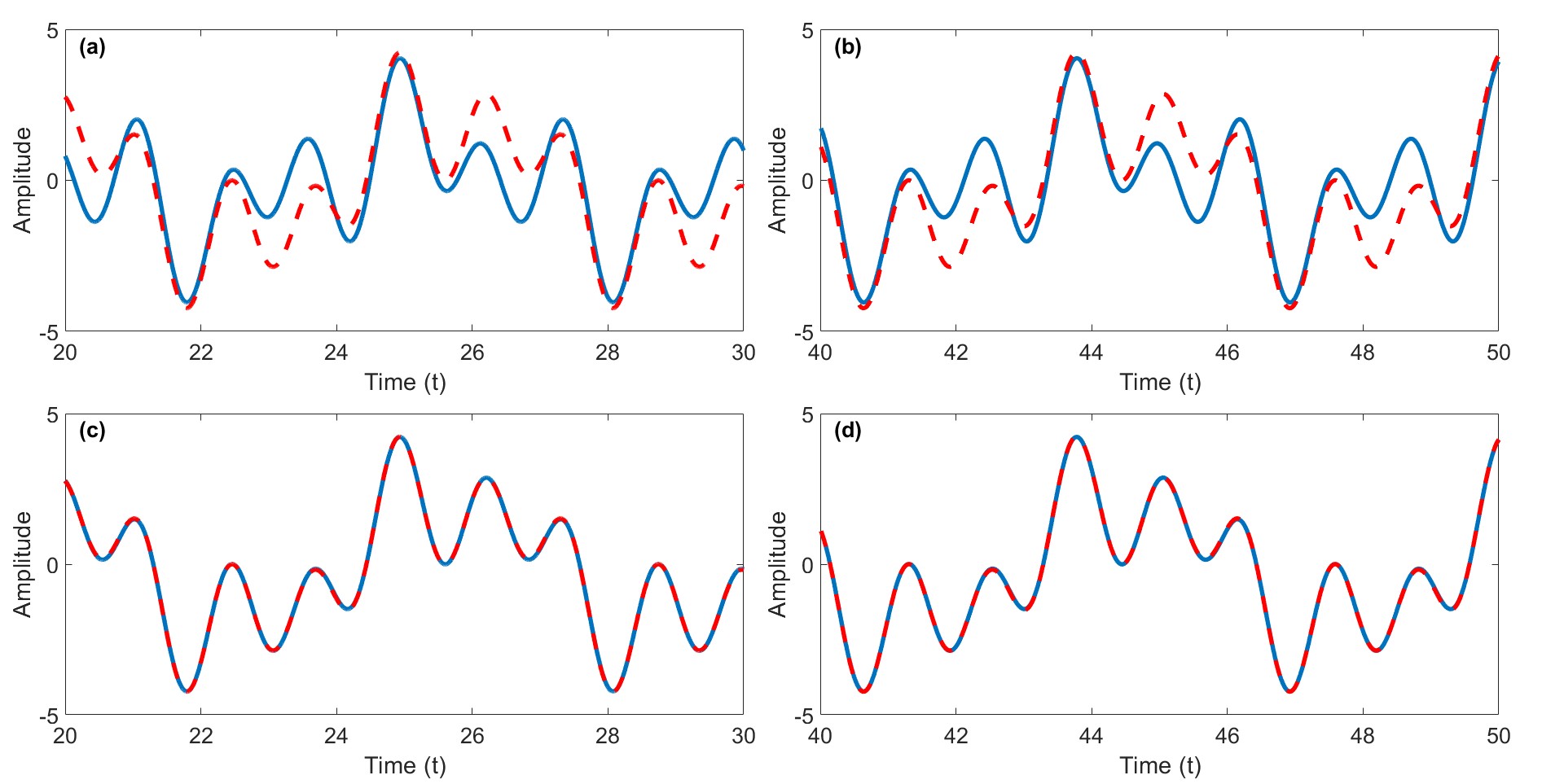}
    \caption{Performance of the 100-node reservoir. Red dashed curves represent the training signal, while blue curves correspond to the best fit produced by the reservoir: (a) Training before optimization, (b) Testing before optimization, (c) Training after optimization, (d) Testing after optimization.}
    \label{example_100_train}
\end{figure}
Figures  \eqref{example_perf} and \eqref{example_100_train} seem to indicate that the proposed optimization has significant potential for improving signal estimation accuracy. 
\begin{table}[h!]
\centering
\begin{tabular}{c c c c c}
\hline
\multicolumn{1}{c}{} & \multicolumn{4}{c}{Average NRMSE} \\ \hline
Network Nodes & \makecell{Training \\ Before Optimization} & \makecell{Training \\ After Optimization} & \makecell{Testing \\ Before Optimization} & \makecell{Testing \\ After Optimization} \\ \hline
10  & 0.7849 & 0.0058 & 0.6943 & 0.0054 \\
100 & 0.5786 & 0.0053 & 0.5197 & 0.0049 \\
200 & 0.5415 & 0.0049 & 0.4983 & 0.0047 \\ \hline
\end{tabular}
\caption{Average training and testing NRMSE before and after optimization for different network sizes. {The NRMSE values presented in the Table are averaged over 50 independent runs.}}
\label{tab:data_table}
\end{table}

The comparison of the average training and testing NRMSE of reservoir computers with 10, 100, and 200 nodes over 50 runs is presented in Table~\eqref{tab:data_table}. The results show that the optimized reservoirs consistently achieve substantially lower NRMSE across all network sizes. Moreover, increasing the number of nodes does not degrade performance; instead, the optimized reservoirs maintain strong accuracy as the number of nodes increases.

The proposed optimal {L}RC performs reliably not only on simple benchmark signals but also on signals derived from a nonlinear system. {As a typical nonlinear system, we use the Lorenz system. The Lorenz system is a set of three ordinary differential equations,
\begin{subequations}
\begin{align}
\dot{x} &= \sigma (y - x),\\
\dot{y} &= x (\rho - z) - y,\\
\dot{z} &= x y - \vartheta z
\end{align} 
\end{subequations}
originally used to model weather patterns but has since been used as a prototypical chaotic system and has been found to model many such systems \cite{chimal2025different}. The parameters $(\sigma = 10, \, \rho = 25, \textrm{and}  \, \vartheta = 8/3)$ are system parameters originally related to the fluid system ($\sigma$ is the Prandtl number, $\rho$ was the Rayleigh number, and $\vartheta$ was related to the fluid layer dimension), which we choose to put the Lorenz system into a chaotic regime. Other simulation parameters are provided in the Supplement.} Figure~\eqref{lorenz_6} demonstrates its performance in an observation task in which the $z$-state of a chaotic Lorenz oscillator is reconstructed from knowledge of its $x$-state. First, the $K=6$ dominant frequencies common to the $x$ and $z$ states are extracted using the method described in the section S4 of the supplement. The reservoir is then trained to estimate $z$ using only the observed $x$ data.

For this task, the non-optimal {L}RC achieves a training NRMSE of $0.9549$, whereas the optimal {L}RC reduces this error to $0.0141$. The test NRMSE follows the same trend: $0.9521$ for the non-optimal {L}RC and $0.0127$ for the optimal {L}RC. These results clearly highlight the significant performance advantage of the optimal {L}RC.

\begin{figure}
    \centering
    \includegraphics[width=1\linewidth]{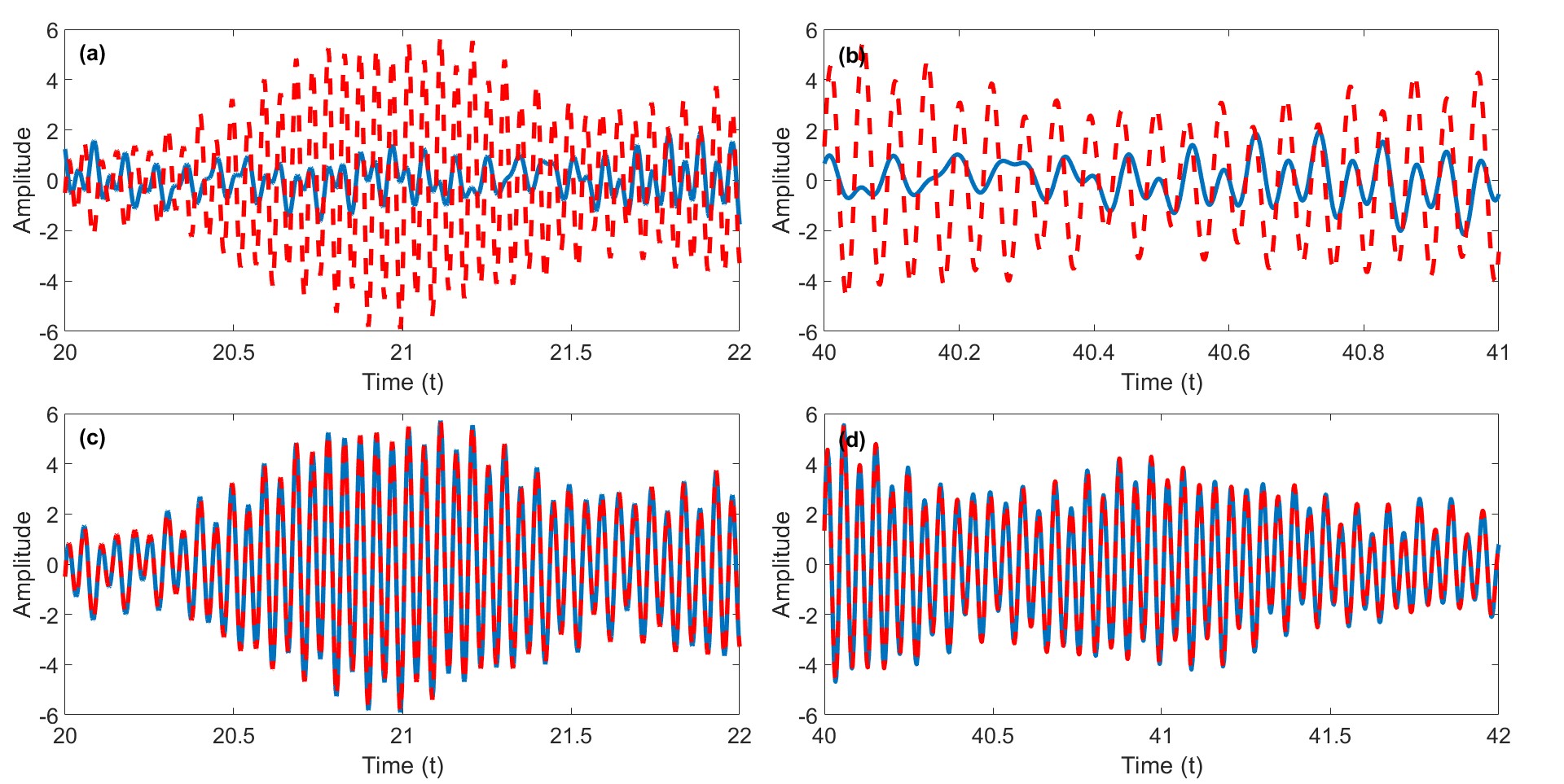}
    \caption{Performance of an $N=25$-node reservoir used to observe the dynamics of the chaotic Lorenz system. Both the input and training signals are restricted to the 
$K=6$ dominant common frequencies. Red dashed curves represent the training signal, while blue curves correspond to the best fit produced by the reservoir: (a) Training before optimization, (b) Testing before optimization, (c) Training after optimization, (d) Testing after optimization.}
    \label{lorenz_6}
\end{figure}

\color{black}
\subsection{Performance Relative to Existing Approaches}

{To evaluate the performance of the proposed optimal LRC, it is compared with two benchmark nonlinear reservoir computers reported in the supplementary material}, using two activation functions: \texttt{tanh} and \texttt{ReLU}, as well as with a linear reservoir computer benchmark. Fig.\ \eqref{comparison} presents the testing NRMSE of the optimized reservoir alongside the three benchmarks, under two scenarios. In the first scenario, the number of reservoir nodes is fixed at $N=50$, while the number $K$ of input frequencies varies; in the second, the number of input frequencies is fixed at $K=10$, while the number of reservoir nodes $N$ varies. Together, these scenarios provide a structured basis for evaluating the relative performance of the proposed optimal LRC across changes in the network size and input complexity.

\begin{figure}
    \centering
    \includegraphics[width=1\linewidth]{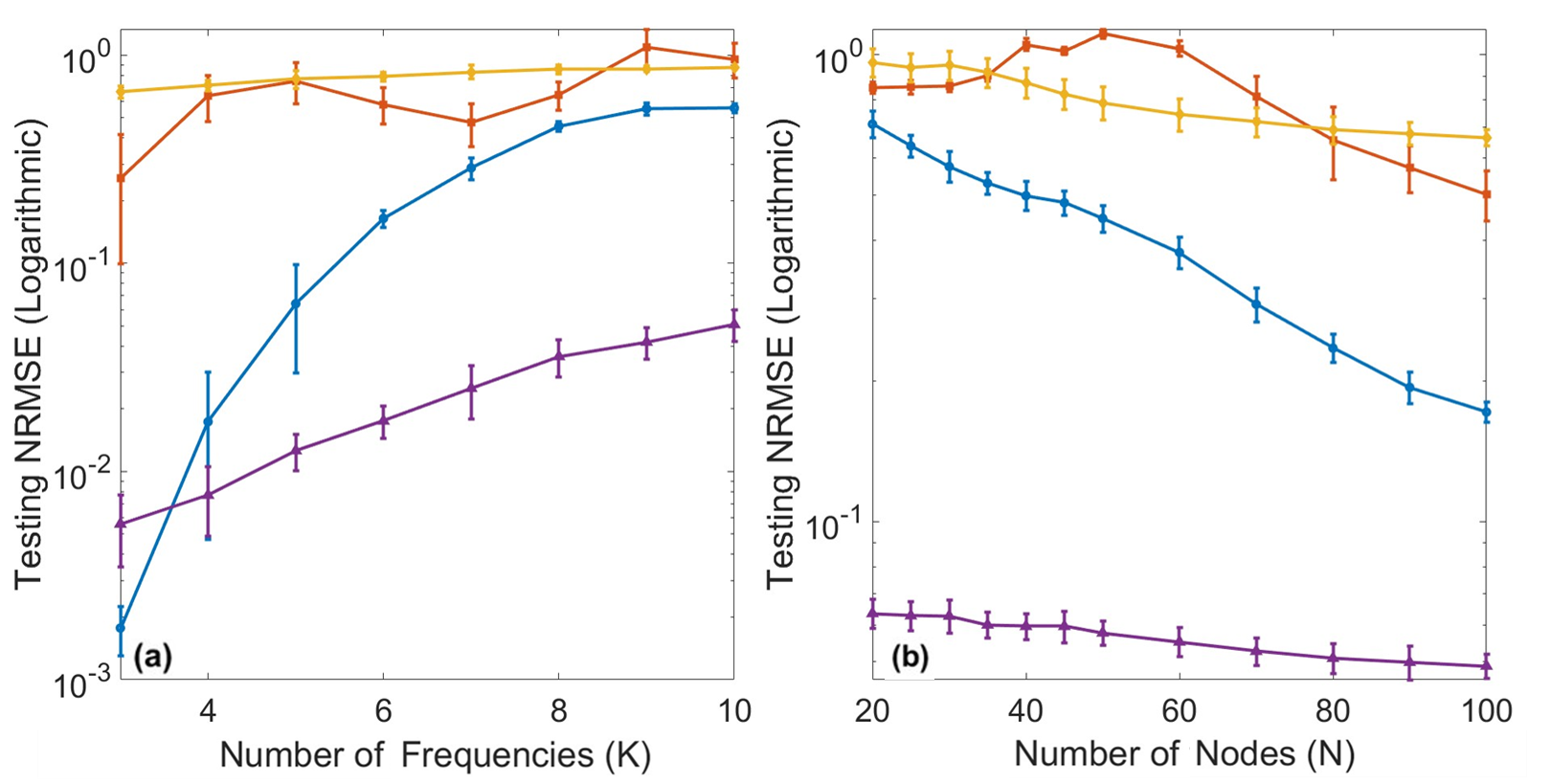}
    \caption{The figures show the average normalized root mean squared testing error (NRMSE) of reservoir computers with their corresponding standard deviations as we vary either the number of input frequencies $K$ or the number of nodes $N$. {Mean and standard deviations are computed from 20 independent runs.} The yellow curve represents the NRMSE of the linear benchmark, the orange curve corresponds to the nonlinear RC with \texttt{ReLU} activation, the blue curve shows the nonlinear RC with \texttt{tanh} activation, and the purple curve indicates the optimized LRC. a) The number of nodes is fixed at $N=50$ while the number of input frequencies varies. b)  The number of input frequencies is fixed at $K=10$ while the number of reservoir nodes varies. The duration of the testing phase is taken to be one third of the duration of the training phase.}
    \label{comparison}
\end{figure}

Figure \eqref{comparison}(a) shows that the 
{lowest testing error is achieved either by the \texttt{tanh}-based nonlinear RC in the case of a small number of frequencies $K$ (blue curve) or by the optimized LRC for larger values of $K$ (purple curve). The linear RC benchmark (yellow curve) and the nonlinear RC with \texttt{ReLU} activation (orange curve) appear to underperform the other two methods over the entire range of the frequencies we have tested.} Figure \eqref{comparison}(b) shows that the 
{optimized LRC (purple curve) achieves the lowest testing error, whereas the \texttt{tanh}-based nonlinear RC (blue curve) shows a greater reduction in error as the number of nodes increases. On the other hand, the linear RC benchmark (yellow curve) and the nonlinear RC with \texttt{ReLU} activation (orange curve) both appear to underperform the other two methods over the entire range of $N$ we have tested. An important observation from Fig.\ \eqref{comparison}(b) is that the optimized LRC achieves significantly lower NRMSE with far fewer nodes compared to the \texttt{tanh}-based RC. This reduction in required network size directly makes the proposed optimized LRC more efficient for real-time and resource-constrained applications.}

{
Standard RCs with a \texttt{tanh} or \texttt{ReLU} reservoir are computationally fast to train, but their performance is limited by the randomly initialized reservoir. In contrast, the optimized linear RC entails a substantially higher computational cost during the optimization phase, as the underlying problem is nonconvex and highly sensitive to the initial guess (see Table~6 in the Supplementary Material). Achieving a near-optimal solution therefore requires evaluating multiple initializations, with the number of required guesses increasing with the network size, which explains the long runtimes observed. For a given task, however, this optimization is performed only once; once the reservoir has been optimized, the associated computational cost does not need to be incurred again.}

{{Despite this,} the optimized RC achieves lower estimation errors and attains a given level of performance with fewer nodes (reducing memory and deployment runtime). Since the optimization is performed offline, the resulting reservoir can be used efficiently in real-time or resource-constrained applications.}

\subsection{Effect of the Cost Function Components on the Optimization Performance}
The optimization formulation, as shown in Eq.\ \eqref{optim}, consists of three terms. The second and third terms have associated weights $\beta_1 \geq 0$ and $\beta_2 \geq 0$, respectively. Therefore, it is important to analyze the impact of these coefficients on the overall optimization process. Figure \eqref{Sensit}  (a) demonstrates the effect of the second and third terms' weights on the optimization of the error. Figure \ref{Sensit} (a), shows that reducing the weight of the second term ($\beta_1$) leads to a decrease in optimization error. This suggests that allowing larger output weights enables the reservoir to more effectively learn the signal, thereby improving performance. Additionally, increasing the weight of the third term ($\beta_2$) also leads to a reduction in error. This is because the third term{, a heuristic algebraic term, } encourages the optimization to produce distinct eigenvalues, which is a desirable property. {We see that $\beta_1$ has a more substantial impact on optimization. Additionally, considering the heuristic role of $\beta_2$, its contribution should only be effective if it is nonzero and not too large. This balance can be seen in the reduction in error near the bottom right of Figure \eqref{Sensit} (a).} 
\subsection{Effect of $\lambda$ on Reservoir Performance}
Performing a sensitivity analysis of the reservoir’s eigenvalues is important for assessing both robustness and generalization. High sensitivity means that small perturbations in eigenvalues can lead to large errors, making the system fragile. It may also indicate overfitting, where the reservoir performs well only under narrowly tuned conditions. Evaluating sensitivity helps ensure stable and reliable performance. Since the optimization formulation focuses on minimizing the error by tuning the reservoir’s eigenvalues and output weights, it is essential to understand how sensitive the reservoir computer is to the perturbations applied to the optimal eigenvalues. In order to perform the sensitivity analysis, a small perturbation  is added to each one of the eigenvalues, 
\begin{equation} \label{pert}
\lambda_{perturbed } = \lambda_{opt}+\rho \delta, 
\end{equation}
where $\delta$ is a number randomly chosen from a Normal distribution and $\rho\geq 0$ is a tunable parameter.
\begin{figure}[ht]
    \centering
    \includegraphics[width=1\linewidth]{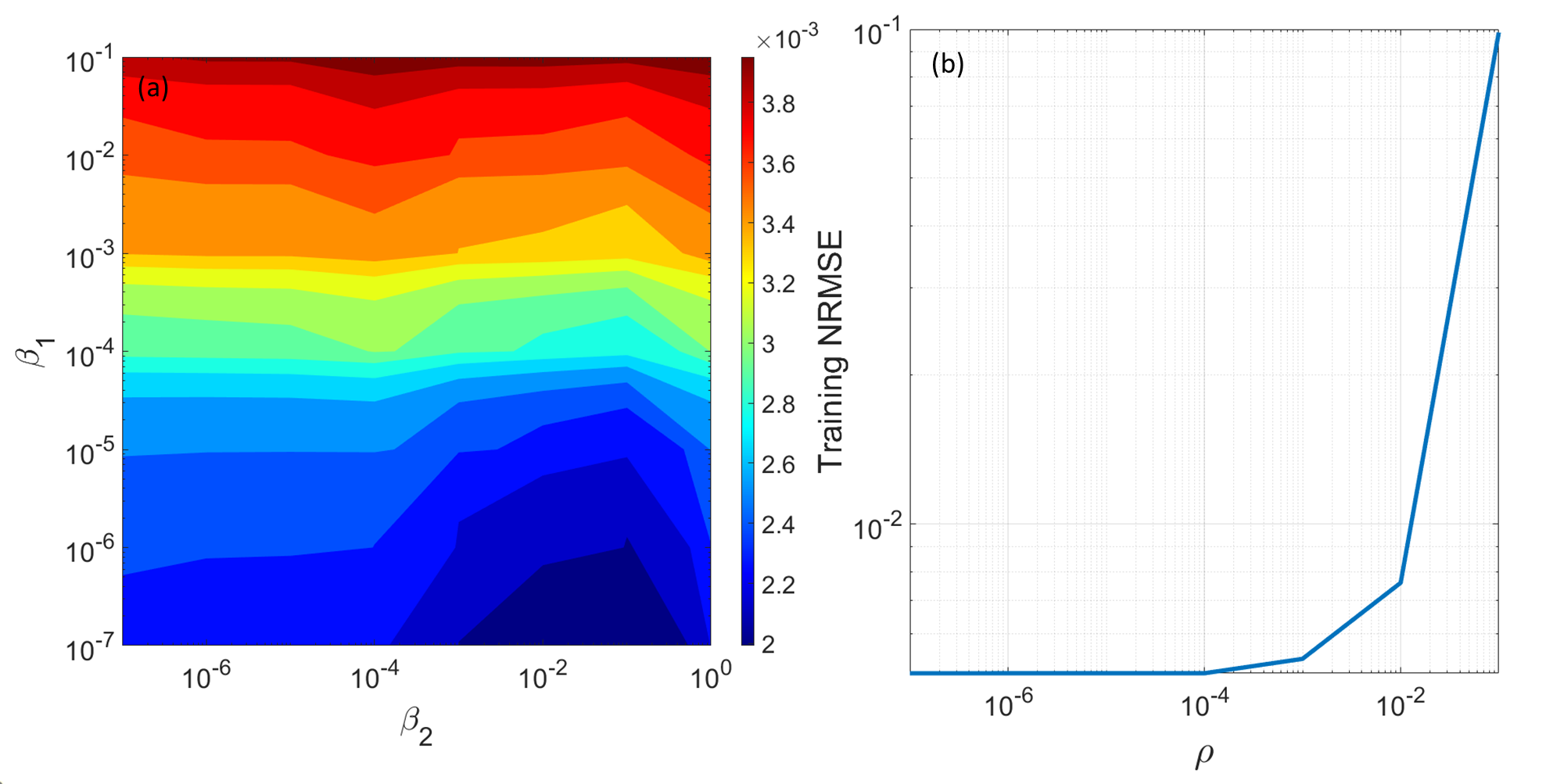}
    \caption{a) Average reservoir training NRMSE as a function of $\beta_1$ and $\beta_2$ for an $N=10$-node reservoir with $\gamma = 5.5$. Each plotted data point, corresponding to a different selection of $\beta_1$ and $\beta_2$, is an average over 10 different runs. b) Sensitivity of the reservoir training NRMSE to perturbations applied to the optimal eigenvalues, based on Eq.\ \eqref{pert}. Each plotted data point is an average over 10 different runs.}
    \label{Sensit}
    \end{figure}
Figure \eqref{Sensit} (b) indicates that the optimal LRC {training} NRMSE  remains largely unaffected by small perturbations applied to the optimal eigenvalues, ie., for $\rho$ that does not exceed a threshold value. Only significant deviations from these optimal values cause noticeable increases in the NRMSE. This behavior suggests that the reservoir is robust and stable around the optimized eigenvalues, making it resilient to minor noise or implementation inaccuracies. Consequently, the optimized reservoir parameters provide a reliable solution, ensuring consistent performance even when slight variations occur in practice.

\subsection{Effect of the Number of the Training Steps $T$ on Reservoir Performance}
A notable advantage of reservoir computing  is its ability to learn from signals of short duration, which reduces the required data length and decreases the overall computational effort during training. Consequently, it is important to evaluate how the duration of the training  phase influences the performance of the optimized reservoir computer. To this end, we have assessed the performance of the optimized LRC using multiple numbers of training steps $T$, and the resulting performance metrics are reported in Table~\ref{tab:my_table}.
The NRMSE values presented in the Table are averaged over 50 independent runs.

\begin{table}[h!] \centering \begin{tabular}{c c c} \hline Training Steps & Avg. Training NRMSE & Avg. Test NRMSE \\ \hline 
500 & 0.0071 & 0.0067 \\ 
1000 & 0.0061 & 0.0059 \\ 
2000 & 0.0059 & 0.0056 \\ 
3000 & 0.0058  & 0.0054 \\ \hline \end{tabular} \caption{Average training and test NRMSE for different numbers of training steps. {The NRMSE values presented in the Table are averaged over 50 independent runs.}} \label{tab:my_table} \end{table}

From Table~\ref{tab:my_table} we see that as the number of training steps $T$ decreases, the NRMSE increases slightly; however, this increase is not significant and can be considered negligible. This indicates that the proposed optimized RC is robust to variations in the duration of the training-phase, while maintaining accurate performance.

\color{black}
\section{Conclusions}
This work makes two
main contributions. First, we demonstrate that the dynamics of a linear RC can be decoupled into a  number of independent modes in a way that {left} the training error 
(and the testing error) unchanged. Second, we show that each of these
modes can be individually optimized for a specific task when the decoupled RC is
reformulated from the time domain into the frequency domain. This reformulation yields a much smaller {state} matrix, since 
the number of input frequencies is typically much smaller than the number of time steps. 
In so doing, we show that the effect of the connectivity of the LRC on the training and testing error is completely captured by the eigenvalues of the reservoir adjacency matrix, which allows us to formulate an optimization problem in the eigenvalues only. 

{A common practice in reservoir computing is to tune the spectral radius by uniformly rescaling all the eigenvalues of the adjacency matrix
\cite{jaeger2001echo}. Our method can be seen as an extension of this technique, which involves fine-tuning each individual eigenvalue.}
Simulations on networks of varying sizes demonstrate that the optimized LRC significantly outperforms randomly constructed reservoirs in both the training and testing phases. 
Moreover, comparisons with established RC benchmarks indicate that the optimized reservoir not only achieves improved performance but also exhibits greater {robustness, scalability, and modularity.} 
{ A limitation of our approach is that the optimization of the eigenvalues is non-convex, making the results highly dependent on the choice of initial guesses and computationally costly. In this work, the initial guesses were chosen randomly, which increases computational time and can lead to sub-optimal results. Identifying a more suitable choice of the initial guess remains an open challenge for future work. Although  not addressed here, our work can be generalized to the case of next generation reservoir computing of Ref.\ \cite{Gauthier2021}}

\color{black}
\section*{Acknowledgement}
We would like to acknowledge generous help from Dr. Saif R. Kazi and Dr. Amir Nazerian. F.S. was supported by grant AFOSR FA9550-24-1-0214.


\begin{thebibliography}{20}
\providecommand{\natexlab}[1]{#1}
\providecommand{\url}[1]{\texttt{#1}}
\expandafter\ifx\csname urlstyle\endcsname\relax
  \providecommand{\doi}[1]{doi: #1}\else
  \providecommand{\doi}{doi: \begingroup \urlstyle{rm}\Url}\fi

\bibitem[Bensoussan et~al.(2022)Bensoussan, Li, Nguyen, Tran, Yam, and Zhou]{BENSOUSSAN2022531}
Alain Bensoussan, Yiqun Li, Dinh Phan~Cao Nguyen, Minh-Binh Tran, Sheung Chi~Phillip Yam, and Xiang Zhou.
\newblock Chapter 16 - machine learning and control theory.
\newblock In Emmanuel Trélat and Enrique Zuazua, editors, \emph{Numerical Control: Part A}, volume~23 of \emph{Handbook of Numerical Analysis}, pages 531--558. Elsevier, 2022.
\newblock \doi{https://doi.org/10.1016/bs.hna.2021.12.016}.
\newblock URL \url{https://www.sciencedirect.com/science/article/pii/S1570865921000314}.

\bibitem[Butschek et~al.(2022)Butschek, Akrout, Dimitriadou, Lupo, Haelterman, and Massar]{Butschek:22}
Lorenz Butschek, Akram Akrout, Evangelia Dimitriadou, Alessandro Lupo, Marc Haelterman, and Serge Massar.
\newblock Photonic reservoir computer based on frequency multiplexing.
\newblock \emph{Opt. Lett.}, 47\penalty0 (4):\penalty0 782--785, Feb 2022.
\newblock \doi{10.1364/OL.451087}.
\newblock URL \url{https://opg.optica.org/ol/abstract.cfm?URI=ol-47-4-782}.

\bibitem[Chimal-Egu{\'\i}a et~al.(2025)Chimal-Egu{\'\i}a, Guzm{\'a}n-Aguilar, Silva-Garc{\'\i}a, B{\'a}ez-Medina, and Cardona-L{\'o}pez]{chimal2025different}
Juan~Carlos Chimal-Egu{\'\i}a, Florencio Guzm{\'a}n-Aguilar, V{\'\i}ctor~Manuel Silva-Garc{\'\i}a, H{\'e}ctor B{\'a}ez-Medina, and Manuel~Alejandro Cardona-L{\'o}pez.
\newblock From different systems to a single common model: A review of dynamical systems leading to lorenz equations.
\newblock \emph{Axioms}, 14\penalty0 (6):\penalty0 465, 2025.

\bibitem[Dale et~al.(2021)Dale, O’Keefe, Sebald, Stepney, and Trefzer]{Dale2021}
Matthew Dale, Simon O’Keefe, Angelika Sebald, Susan Stepney, and Martin~A. Trefzer.
\newblock Reservoir computing quality: connectivity and topology.
\newblock \emph{Natural Computing}, 20\penalty0 (2):\penalty0 205--216, 2021.
\newblock ISSN 1572-9796.
\newblock \doi{10.1007/s11047-020-09823-1}.
\newblock URL \url{https://doi.org/10.1007/s11047-020-09823-1}.

\bibitem[Del~Frate et~al.(2021)Del~Frate, Shirin, and Sorrentino]{del2021reservoir}
Enrico Del~Frate, Afroza Shirin, and Francesco Sorrentino.
\newblock Reservoir computing with random and optimized time-shifts.
\newblock \emph{Chaos: An Interdisciplinary Journal of Nonlinear Science}, 31\penalty0 (12), 2021.

\bibitem[Gauthier et~al.(2021)Gauthier, Bollt, Griffith, and Barbosa]{Gauthier2021}
Daniel~J. Gauthier, Erik Bollt, Aaron Griffith, and Wendson A.~S. Barbosa.
\newblock Next generation reservoir computing.
\newblock \emph{Nature Communications}, 12\penalty0 (1):\penalty0 5564, September 2021.
\newblock ISSN 2041-1723.
\newblock \doi{10.1038/s41467-021-25801-2}.
\newblock URL \url{https://doi.org/10.1038/s41467-021-25801-2}.

\bibitem[Hashemi et~al.(2023)Hashemi, Orzechowski, Mikkola, and McPhee]{Hashemi2023}
Arash Hashemi, Grzegorz Orzechowski, Aki Mikkola, and John McPhee.
\newblock Multibody dynamics and control using machine learning.
\newblock \emph{Multibody System Dynamics}, 58\penalty0 (3):\penalty0 397--431, 2023.
\newblock ISSN 1573-272X.
\newblock \doi{10.1007/s11044-023-09884-x}.
\newblock URL \url{https://doi.org/10.1007/s11044-023-09884-x}.

\bibitem[Jaeger(2001)]{jaeger2001echo}
Herbert Jaeger.
\newblock The “echo state” approach to analysing and training recurrent neural networks.
\newblock Technical Report GMD Report 148, German National Research Center for Information Technology, 2001.

\bibitem[Jordanou et~al.(2022)Jordanou, Antonelo, and Camponogara]{9664461}
Jean~Panaioti Jordanou, Eric~Aislan Antonelo, and Eduardo Camponogara.
\newblock Echo state networks for practical nonlinear model predictive control of unknown dynamic systems.
\newblock \emph{IEEE Transactions on Neural Networks and Learning Systems}, 33\penalty0 (6):\penalty0 2615--2629, 2022.
\newblock \doi{10.1109/TNNLS.2021.3136357}.

\bibitem[Lu et~al.(2017)Lu, Pathak, Hunt, Girvan, Brockett, and Ott]{10.1063/1.4979665}
Zhixin Lu, Jaideep Pathak, Brian Hunt, Michelle Girvan, Roger Brockett, and Edward Ott.
\newblock Reservoir observers: Model-free inference of unmeasured variables in chaotic systems.
\newblock \emph{Chaos: An Interdisciplinary Journal of Nonlinear Science}, 27\penalty0 (4):\penalty0 041102, 04 2017.
\newblock ISSN 1054-1500.
\newblock \doi{10.1063/1.4979665}.
\newblock URL \url{https://doi.org/10.1063/1.4979665}.

\bibitem[Lukoševičius and Jaeger(2009)]{lukosevicius2009reservoir}
Mantas Lukoševičius and Herbert Jaeger.
\newblock Reservoir computing approaches to recurrent neural network training.
\newblock \emph{Computer Science Review}, 3\penalty0 (3):\penalty0 127--149, 2009.

\bibitem[McCreesh and Cortés(2024)]{10659224}
Michael McCreesh and Jorge Cortés.
\newblock Control of linear-threshold brain networks via reservoir computing.
\newblock \emph{IEEE Open Journal of Control Systems}, 3:\penalty0 325--341, 2024.
\newblock \doi{10.1109/OJCSYS.2024.3451889}.

\bibitem[Metzner et~al.(2025)Metzner, Schilling, Maier, and Krauss]{10.1162/neco_a_01770}
Claus Metzner, Achim Schilling, Andreas Maier, and Patrick Krauss.
\newblock Nonlinear neural dynamics and classification accuracy in reservoir computing.
\newblock \emph{Neural Computation}, 37\penalty0 (8):\penalty0 1469--1504, 07 2025.
\newblock ISSN 0899-7667.
\newblock \doi{10.1162/neco_a_01770}.
\newblock URL \url{https://doi.org/10.1162/neco\_a\_01770}.

\bibitem[Nathe et~al.(2023)Nathe, Pappu, Mecholsky, Hart, Carroll, and Sorrentino]{nathe2023reservoir}
Chad Nathe, Chandra Pappu, Nicholas~A Mecholsky, Joe Hart, Thomas Carroll, and Francesco Sorrentino.
\newblock Reservoir computing with noise.
\newblock \emph{Chaos: An Interdisciplinary Journal of Nonlinear Science}, 33\penalty0 (4), 2023.

\bibitem[Pathak et~al.(2018)Pathak, Hunt, Girvan, Lu, and Ott]{PhysRevLett.120.024102}
Jaideep Pathak, Brian Hunt, Michelle Girvan, Zhixin Lu, and Edward Ott.
\newblock Model-free prediction of large spatiotemporally chaotic systems from data: A reservoir computing approach.
\newblock \emph{Phys. Rev. Lett.}, 120:\penalty0 024102, Jan 2018.
\newblock \doi{10.1103/PhysRevLett.120.024102}.
\newblock URL \url{https://link.aps.org/doi/10.1103/PhysRevLett.120.024102}.

\bibitem[Shen et~al.(2025)Shen, Miyazaki, and Kawashima]{10816480}
Junyi Shen, Tetsuro Miyazaki, and Kenji Kawashima.
\newblock Control pneumatic soft bending actuator with feedforward hysteresis compensation by pneumatic physical reservoir computing.
\newblock \emph{IEEE Robotics and Automation Letters}, 10\penalty0 (2):\penalty0 1664--1671, 2025.
\newblock \doi{10.1109/LRA.2024.3523229}.

\bibitem[Sugiura et~al.(2024)Sugiura, Ariizumi, Asai, and Azuma]{10221724}
Shuhei Sugiura, Ryo Ariizumi, Toru Asai, and Shun-Ichi Azuma.
\newblock Nonessentiality of reservoir’s fading memory for universality of reservoir computing.
\newblock \emph{IEEE Transactions on Neural Networks and Learning Systems}, 35\penalty0 (11):\penalty0 16801--16815, 2024.
\newblock \doi{10.1109/TNNLS.2023.3298013}.

\bibitem[Wyffels et~al.(2008)Wyffels, Schrauwen, and Stroobandt]{10.1007/978-3-540-87536-9_83}
Francis Wyffels, Benjamin Schrauwen, and Dirk Stroobandt.
\newblock \emph{Stable Output Feedback in Reservoir Computing Using Ridge Regression}.
\newblock Springer Berlin Heidelberg, Berlin, Heidelberg, 2008.
\newblock ISBN 978-3-540-87536-9.

\bibitem[Yan et~al.(2024)Yan, Huang, Bienstman, Tino, Lin, and Sun]{Yan2024}
Min Yan, Can Huang, Peter Bienstman, Peter Tino, Wei Lin, and Jie Sun.
\newblock Emerging opportunities and challenges for the future of reservoir computing.
\newblock \emph{Nature Communications}, 15\penalty0 (1):\penalty0 2056, 2024.
\newblock ISSN 2041-1723.
\newblock \doi{10.1038/s41467-024-45187-1}.
\newblock URL \url{https://doi.org/10.1038/s41467-024-45187-1}.

\bibitem[Zhao et~al.(2024)Zhao, Sun, Li, Jia, and Xu]{Zhao_2025}
Xiaoning Zhao, Yougang Sun, Yanmin Li, Ning Jia, and Junqi Xu.
\newblock Applications of machine learning in real-time control systems: a review.
\newblock \emph{Measurement Science and Technology}, 36\penalty0 (1):\penalty0 012003, nov 2024.
\newblock \doi{10.1088/1361-6501/ad8947}.
\newblock URL \url{https://dx.doi.org/10.1088/1361-6501/ad8947}.

\end{thebibliography}
\end{document}